\documentclass[onecolumn,english]{sigcomm-alternate}
\usepackage[latin9]{inputenc}
\usepackage{babel}
\usepackage{array}
\usepackage{rotating}
\usepackage{pifont}
\usepackage{float}
\usepackage{units}
\usepackage{multirow}
\usepackage{amsmath}
\usepackage{graphicx}
\PassOptionsToPackage{normalem}{ulem}
\usepackage{ulem}
\usepackage{lipsum}
\usepackage[hyperfootnotes=false,unicode=true,bookmarks=false,breaklinks=false,pdfborder={0 0 1},backref=false,colorlinks=false]{hyperref}

\newcommand\blfootnote[1]{%
  \begingroup
  \renewcommand\thefootnote{}\footnote{#1}%
  \addtocounter{footnote}{-1}%
  \endgroup
}
\makeatletter

\providecommand{\tabularnewline}{\\}
\floatstyle{ruled}
\newfloat{algorithm}{tbp}{loa}
\providecommand{\algorithmname}{Algorithm}
\floatname{algorithm}{\protect\algorithmname}

\usepackage{calc}
\usepackage{multirow}
\newsavebox\zzz
\def\mystrut{%
\dimen@\wd\zzz
\divide\dimen@\thr@@
\advance\dimen@-\dp\@arstrutbox
\rule\z@\dimen@}
\def\rotatezzz{%
\rotatebox{90}{\rlap{\kern-\dp\@arstrutbox\usebox\zzz}}}

\usepackage{balance}
\usepackage{cite}
\usepackage{breakurl}
\usepackage{soul}

\usepackage{microtype}
\usepackage[font={small}]{caption}

\newtheorem{remark}{Remark}
\newtheorem{lemma}{Lemma}

\newcommand{\cref}[1]{Chapter~\ref{#1}}

\newcommand{\first}{\emph{(i)}~}
\newcommand{\second}{\emph{(ii)}~}
\newcommand{\third}{\emph{(iii)}~}
\newcommand{\fourth}{\emph{(iv)}~}

\newcommand{\ie}{i.e., \@}
\newcommand{\eg}{e.g., \@}

\newcommand{\etal}{et~al.}

\newcounter{fn1}
\newcounter{fn2}
\newcounter{fn3}
\newcounter{fn4}
\newcounter{fn5}
\makeatother

\begin{document}

\title{On the Interplay of Link-Flooding Attacks and Traffic Engineering}

\numberofauthors{3}
\author{
% 1st. author
Dimitrios Gkounis\\
       \affaddr{NEC Labs Europe, Germany}\\
       \email{dimitrios.gkounis@neclab.eu}
% 2nd. author
\alignauthor
Vasileios Kotronis\\
       \affaddr{ETH Zurich, Switzerland}\\
       \email{biece89@gmail.com}
% 3rd. author
\alignauthor
Christos Liaskos\\
       \affaddr{FORTH, Greece}\\
       \email{cliaskos@ics.forth.gr}
\and  % use '\and' if you need 'another row' of author names
% 4th. author
\alignauthor Xenofontas Dimitropoulos\\
       \affaddr{FORTH, Greece}\\
       \email{fontas@ics.forth.gr}
}

\maketitle
\begin{abstract}
Link-flooding attacks have the potential to disconnect even entire
countries from the Internet. Moreover, newly proposed \emph{indirect}
link-flooding attacks, such as ``Crossfire'', are extremely hard
to expose and, subsequently, mitigate effectively. Traffic Engineering
(TE) is the network's natural way of mitigating link overload events,
balancing the load and restoring connectivity. This work poses the
question: \emph{Do we need a new kind of TE to expose an attack as
well?} The key idea is that a carefully crafted, attack-aware TE could
force the attacker to follow improbable traffic patterns, revealing
his target and his identity over time. We show that both existing
and novel TE modules can efficiently expose the attack, and study
the benefits of each approach. We implement defense prototypes using
simulation mechanisms and evaluate them extensively on multiple real
topologies.
\end{abstract}
%\category{C.2.0}{Computer Communication Networks}{General-security and protection} \category{C.2.3}{Computer Communication Networks}{Network Operations-network management}

\keywords{DDoS defense; link-flooding attack; traffic engineering.}

\section{Introduction\label{sec:intro}}
\blfootnote{This work was funded by the European Research Council, grant EU338402, project \href{http://www.netvolution.eu}{NetVOLUTION}.}
Some of the most powerful DDoS (Distributed Denial of Service) attacks
ever have been observed during 2013 and 2014, reaching traffic rates
greater than $300$~$Gbps$~\cite{spamhaus}. Moreover, new types
of indirect DDoS link-flooding attacks have been recently proposed,
which are extremely difficult to mitigate~\cite{crossfire,Coremelt}.
In particular, the \textit{Crossfire} attack~\cite{crossfire} (illustrated
in Fig.~\ref{figRationale}) seeks to cut-off a given network area
(target) from the Internet, while sending \emph{no} attack traffic
directly to the target as follows: \first the attacker detects the
link-map around the target by executing traceroutes towards many points
within the network, \second locates critical links that connect the
intended target to the Internet, \third deduces non-targeted network
areas (decoys) that are also served via the critical links, and \fourth
consumes the bandwidth of the critical links with multiple low-bandwidth
flows (\eg normal HTTP messages) originating from attacker-controlled
bots and destined towards the decoys. Thus, the target loses Internet
connectivity, without noticing the attacker's traffic.

Traffic Engineering (TE) is the expected reaction of a network on
link-overload events. Regardless of their cause, TE is triggered to
perform re-routing and load-balance the network traffic. Thus, there
is a natural interplay between the Crossfire attack and TE. Following
a TE round, the attack will pinpoint and flood new links and the cycle
repeats, potentially using different decoys and bots~\cite{crossfire}.
Despite this core-placement of TE in the attack cycle, its effects
on the attack exposure potential remain unknown. Nonetheless, even
existing TE schemes could potentially expose the attack, without alterations.
For instance, an administrator could exploit the repeated TE, and
monitor network areas that are persistently affected by link-flooding
events. Such areas can be marked as probable targets, gradually implying
that an attack is in progress. Additionally, traffic sources that
persistently react to re-routing can be monitored for the purpose
of affecting a specific target. If certain sources are recorded several
times in links that are DoS'ed, effectively behaving as a bot swarm,
they can be marked as suspicious.

The present work studies the effect of TE on the attack exposure via
analysis and experimentation. The focus is on the destination-based
routing case, which remains widely adopted \cite{caesar2005bgp}.
Our contribution is two-fold, showing that: \first Crossfire attacks
can be exposed in a manner agnostic to the underlying, attack-unaware
TE scheme, and \second certain attack-aware TE schemes can contain
the effects of the attack within a small part of the network, whereas
attack-unaware TE may cause network-wide routing changes. Nonetheless,
attack containment can increase the attack exposure time, while optimizing
this trade-off is an interesting open problem. \vspace{-8bp}

\section{System Model\label{sec:lf-attacks}}

The study assumes a destination-based routed network and the Crossfire
attack model \cite{crossfire}. The network can have multiple gateways
to the Internet, while each network node represents an area that may
contain one or more physical machines. The network contains a node
that represents the \emph{target} area of a bot swarm. Bots are entities
with unique identifiers (e.g., machines with different IPs) that operate
from beyond or within the studied network. The attack model defines
a continuous cycle of interactions between the \emph{attacker} (bot
swarm) and the \emph{defender} (network administrator)~\cite{crossfire}.

\textbf{\uline{Attacker Model}}\textbf{:} The attacker seeks to
cut-off the paths connecting the gateways to the target. On these
paths there exist public servers, the \textit{Decoy Servers}, and
the respective \textit{Critical Links}, which lead to both the target
and the decoy servers. The attacker first constructs a map of the
network links (the \textit{link-map}) around the target, e.g., by
executing \emph{traceroutes} towards multiple points in the network~\cite{gkounis}.
Then, he floods critical links by sending traffic only to decoy servers.
Thus, all paths connecting the target to the gateway(s) are blocked
due to congestion at some links, while the target sees no apparent
cause, such as a high volume of direct traffic.

\textbf{\uline{Defender Model}}\textbf{:} The goal of the defender
is to keep the network running without any severe performance degradation
(\eg flooded links) and to expose probable targets and bots. Therefore,
he: \first monitors traffic load on his network and detects links
that are flooded,~\second balances the incoming traffic by re-routing
load destined to different destinations (including the target, decoys,
etc.), and~\third executes actions to expose probable attackers
and their target.

\textbf{\uline{Attacker - Defender Interaction}}\textbf{:} The
attacker monitors the network routes and reacts to routing changes
performed by the defender, as exemplary shown in Fig.~\ref{figRationale}.
Bots will then change their decoy server selection in case the re-routing
has diverted their load from the critical link(s), repeating the cycle.
Thus, the attacker updates the link-map, recalculates the critical
links and floods them again with several bot-originated flows. In
the example of Fig.~\ref{figRationale}, both the links and the decoys
vary per cycle. The attacker may also vary the bots that participate
in each attack cycle, in an effort to make their presence less persistent.
In this case, the defender can reply by simply distributing traffic
to the target across multiple routes, making future flooding attempts
harder.
\begin{figure}
\begin{centering}
\includegraphics[bb=0bp 0bp 580bp 250bp,clip,width=0.95\columnwidth]{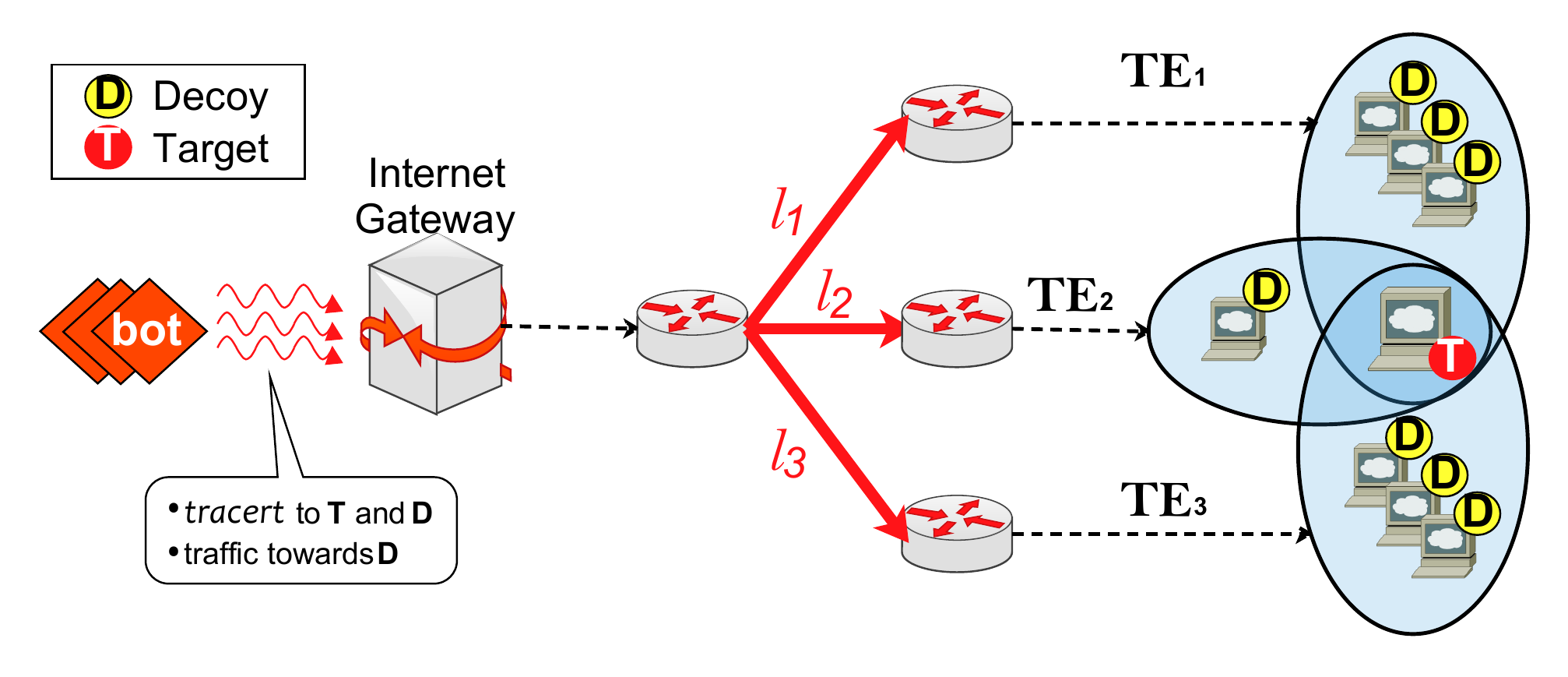}
\par\end{centering}

\caption{\label{figRationale}Attacker-Defender interaction. The attacker floods
a link~$l_{1}$. The defender then re-routes traffic~(TE$_{2}$).
The attacker updates the selected decoy servers, flooding link~$l_{2}$.
The defender replies with TE$_{3}$ and the attacker floods link $l_{3}$,
and so on.\vspace{-8bp}
}
\end{figure}

\section{Analysis}

\label{sec:defense}

The defense objectives are assumed to be the following: \first expose
probable targets and, hence, an attack, and \second expose traffic
sources which are consistently involved in it.

The exposure of probable attack targets is considered as the primary
goal, and it is based on the fact that areas affected by malevolent
link-floods contain the attack target. Thus, as shown in Fig. \ref{figRationale},
if the intersection of the affected areas is not empty, then there
must be an attack towards a target in this area intersection.

The exposure of bots relies on the assumption that malevolent traffic
sources (bots) tend to: \first change their destinations, or \second
open new connections, persistently affecting a probable target (Fig.
\ref{figRationale}). On the contrary, benign causes of link-floods
(\eg flash-crowds) do not re-adjust to routing changes in such a
persistent manner~\cite{xue2014towards}.

Nonetheless, studies have proven that bot detection approaches based
on the bot reuse assumption does not guarantee detection~\cite{crossfire,INFOCOM}.
For instance, if the botnet size is vast, the attacker needs not use
each bot frequently. In such cases, the defender can raise his bot
detection ratio by deploying multiple, generic heuristics in parallel,
such as reflector nets, white holes, bot-hunters, and rate limiters~\cite{shin2013fresco}.
The attacker should then avoid all deployed traps on a regular basis
and for every bot, which can be extremely challenging. Thus, the proposed
scheme also considers the smooth collaboration with other third-party
defense mechanisms in order to increase the chances of bot detection.

In light of these remarks, we operate as follows. Section \ref{sec:prereq}
presents the preliminaries. The defense workflow is modeled and analyzed
in Section~\ref{sec:def-workflow}, assuming independence from the
underlying TE module. In Section~\ref{sec:remote}, we study the
case of attack-aware TE modules. Collaboration with third-party defense
mechanisms is discussed in Section~\ref{sub:Collaboration}.

\subsection{Notation and Assumptions\label{sec:prereq}}

\textbf{\uline{Notation}}\textbf{.} The network is denoted as $W=\left\{ \mathcal{N},\mathcal{L}\right\} $
comprising a set of nodes $\mathcal{N}$ and links $\mathcal{L}$.
A single node is denoted as $n\in\mathcal{N}$ and a directed network
link as $l\in\mathcal{L}$. $s(l)$~and~$d(l)$ represent the source
and destination nodes of~$l$. The network has a set $G$ of gateway
nodes $g\in G$ and a persistently targeted node $target\in\mathcal{N}$.
The path connecting a node $n$ to a destination $m$ is expressed
as $\overrightarrow{p_{(n,m)}}=\left\{ l_{1},l_{2},l_{3},\ldots,l_{k}\right\} $,
where $\left\{ \ldots\right\} $ are the ordered links comprising
the path, with $s(l_{1})=n$, $d(l_{k})=m$. The \emph{free} bandwidth
of link $l$ is~$\mathcal{B}(l)$.

Link-flooding attacks deplete the bandwidth of certain network links,
using bots to send traffic over the gateways to the network. The defender
deduces if a link~$l$ is flooded, \ie $flooded(l)\to\nicefrac{true}{false}$,
based on existing approaches (\eg~\cite{xue2014towards}). We define
$\mathbf{Src}\left\{ l\right\} $ as the set of IP addresses, \ie
malevolent or legitimate data sources that contribute to the traffic
flowing within link $l$. Likewise, $\mathbf{Dst}\left\{ l\right\} $
denotes the nodes served \emph{via} link $l$ in the examined down-stream
direction, \ie $\mathbf{Dst}\left\{ l\right\} =\left\{ n:\,l\in\overrightarrow{p_{(g,n)}}\right\} $.

\textbf{\uline{Routing Changes}}\textbf{.} We next define $\mathcal{R}_{TE}$
as the set that contains all routing configurations that can be deployed
in~$W$. The instance of routing rules presently adopted within the
network is denoted as $r_{TE}\in\mathcal{R}_{TE}$. Migrating from
$r_{TE}$ to another configuration $r'_{TE}$ requires the addition/removal
of rules at the routing tables $T_{n},\,n\in\mathcal{N}$, of the
network nodes that route/switch traffic. In order to denote the rule
modifications needed for this migration, we define the associative
operator $\ominus$ on two routing instances, $r_{TE}$ and $r'_{TE}$,
which counts the required routing table changes as~\cite{Thaler.2000e}:
\begin{equation}
r_{TE}\ominus r_{TE}^{*}=\sum_{\forall n\in\mathcal{N}}\left\Vert T_{n}^{r_{TE}}\cup T_{n}^{r'_{TE}}-T_{n}^{r_{TE}}\cap T_{n}^{r'_{TE}}\right\Vert
\end{equation}
where $\cap,\,\cup,\,-$ are the set intersection, union and difference
operators, and $\left\Vert *\right\Vert $ is the cardinality of a
set $*$. The migration from $r_{TE}$ to $r'_{TE}$ is done with
existing approaches~\cite{jin2014dynamic}.

\textbf{\uline{Attack}}\textbf{.} The attacker first discovers
the paths $\overrightarrow{p_{(g,target)}}$, $g\in G$, for the current
routing tree, $r_{TE}$. A practical procedure based on distributed
traceroutes is detailed in~\cite{gkounis}. Then, for each path,
he selects and attacks the link with the smallest amount of free bandwidth,
seeking to limit the number of required bots~\cite{crossfire}. In
the \emph{traceroutes} approach, link loads can be inferred via the
corresponding delays. Increased link delay indicates high queuing
delay and, therefore, high load as well. In addition, the selected
link should be on the path to some minimum number of decoy nodes,
$D_{min}$~\cite{crossfire}. More decoys mean that the attacker
can mask his traffic more effectively as legitimate, since each destination
will receive less bot traffic. Thus, the critical links are selected
as:

$l=argmin_{\left(l\right)}\left\{ \mathcal{B}(l):\,l\in\overrightarrow{p_{(g,target)}},\,\left\Vert \mathbf{Dst}\left\{ l\right\} \right\Vert \ge D_{min}\right\} $\\
Finally, he selects the bots with the lowest participation in past
attack steps and launches the new attack.

\subsection{Defense Analysis and Workflow\label{sec:def-workflow}}

Let $t=0,1,\ldots$ denote the time moments when the network administrator
notices link-flooding events. Let the flooded links at time $t$ be
$l_{i}^{t},\,i=0,1,\ldots$. The set of nodes whose connectivity is
affected by these flooded links is:
\begin{equation}
\mathbf{Dst}^{t}=\sideset{}{_{\forall i}}\bigcup\mathbf{Dst}\left\{ l_{i}^{t}\right\}
\end{equation}
The TE module is naturally called to load-balance the traffic, relieving
the congested links and restoring connectivity. In terms of attack
target exposure, it would be desirable to apply a new routing tree,
$r_{TE}^{t+1}$, that decouples each node $n\in\mathbf{Dst}^{t}$
from the rest of the set $\mathbf{Dst}^{t}$. In Fig.~\ref{figRationale},
for example, the attack at cycle 1 affects 3 decoys and the intended
target. Thus, the defender sees 4 probable targets. If the subsequent
TE assigned dedicated, link-disjoint paths to each of the four nodes,
the attacker would not affect the 3 original decoys again. Thus, the
real target would stand out. In general, the $r_{TE}^{t+1}$ that
connects each node $n\in\mathbf{Dst}^{t}$ to the network gateways
via link-disjoint paths is:
\begin{equation}
r_{TE}^{t+1}:\,\sideset{}{_{\forall n\in\mathbf{N}^{t}}}\bigcap\overrightarrow{p_{(G,n)}^{t+1}}=\emptyset\label{EQOPTIM}
\end{equation}
Based on condition (\ref{EQOPTIM}), we observe:

\begin{remark}\label{OptimalTE}Topological path-diversity favors
the exposure of Crossfire attack targets.\end{remark}

Condition (\ref{EQOPTIM}) and Remark \ref{OptimalTE} are generally
aligned to the load-balancing objective of TE. Relieving flooded links
means that certain nodes need to be served via decoupled paths, while
path-diversity is known to favor the efficiency of TE~\cite{Hopps.2000b,Buriol.2005}.
Nonetheless, in terms of attack exposure, $r_{TE}^{t+1}$ should also
decouple $\mathbf{Dst}^{t}$ from all past sets $\mathbf{Dst}^{\tau},\,\tau=0\ldots t-1$:
\begin{equation}
r_{TE}^{t+1}:\,\left\Vert \sideset{}{_{\tau=0}^{t+1}}\bigcap\mathbf{Dst}^{\tau}\right\Vert =1\label{EQOPTIMALt}
\end{equation}
Condition (\ref{EQOPTIMALt}) also incorporates the concern that no
past $r_{TE}$ should be repeated at $t+1$, in order to avoid loops,
which benefit the attacker. Thus, the paths available to the TE process
will reduce with $t$, eventually hindering its load-balancing potential.

Due to this limitation, we initially study the case where TE and defense
are decoupled. This constitutes the common, real-world case, where
TE is executed just for load-balancing, without Crossfire-derived
path restrictions or memory~\cite{Hopps.2000b,Buriol.2005}. Then,
assuming sufficient path-diversity for efficient load balancing, we
study the probability of a network node being coupled to the attack
target \emph{by chance}.
\begin{algorithm}[t]
{\small{}~1:}\textbf{\small{}~On init do }{\small{}//Applies at
initialization as well.}{\small \par}

{\small{}~2:~\ \ \ $suspect_{IPs}\gets\emptyset;$ \ \ }\emph{\small{}//$Hash[IP]\to penalty$.}{\small \par}

{\small{}~3:\ \ \ $suspect_{targets}\gets\emptyset;$ \ \ }\emph{\small{}//$Hash[node]\to penalty$.}{\small \par}

{\small{}~4:\ \ \ $r_{TE}\gets$any from $\mathcal{R}_{TE};$
\ \ \ }\emph{\small{}//}{\small{}Initial TE routing rules}\emph{\small{}.}{\small \par}

{\small{}~5:~}\textbf{\small{}On event}{\small{} $flooded(l)$ }\textbf{\small{}do}{\small \par}

{\small{}~6:\ \ \ }{\small{}\uline{async\_call}}{\small{}:
$Penalize(\mathbf{Src}\left\{ l\right\} ,\,\mathbf{Dst}\left\{ l\right\} )$\ \ \ }\emph{\small{}//}{\small{}Non-blocking
call.}{\small \par}

{\small{}~7:\ \ \ $r_{TE}\gets r'_{TE}\in\mathcal{R}_{TE};$ \ \ \ }\emph{\small{}//}{\small{}Execute
new TE routing}\emph{\small{}.}{\small \par}

{\small{}~8:~}\textbf{\small{}Procedure}{\small{} $Penalize(\mathbf{Src},\,\mathbf{Dst})$
//Reinforcement learning.}{\small \par}

{\small{}~9:\ \ \ $suspect_{IPs}\gets penalize\_IPs(suspect_{IPs},\,\mathbf{Src});$}{\small \par}

{\small{}10:\ \ \ $suspect_{targets}\gets penalize\_nodes\left(suspect_{targets},\,\mathbf{Dst}\right);$}{\small \par}

{\small{}11:\ \ \ }\textbf{\small{}if}{\small{} $penalization\_conclusive\left(suspect_{IPs},\,suspect_{targets}\right)$}{\small \par}

{\small{}12:\ \ \ \ \ Deduce the attack target from $suspect_{targets}$;}{\small \par}

{\small{}13:\ \ \ \ \ Deduce bots using $suspect_{IPs}$;}{\small \par}

{\small{}14:\ \ \ \ \ $suspect_{IPs}\gets\emptyset;$}{\small \par}

{\small{}15:\ \ \ \ \ $suspect_{targets}\gets\emptyset;$}{\small \par}

{\small{}16:\ \ \ }\textbf{\small{}end if}{\small \par}

{\small{}\caption{{\small{}\label{ALG}Defense Workflow against Crossfire.}}
}{\small \par}
\end{algorithm}

\begin{lemma}\label{Lemma}The probability of a non-targeted node
being coupled to the target at any attack step is bounded in $\left[0,\nicefrac{1}{2}\right]$.\end{lemma}

\begin{proof}See Appendix.\end{proof}

Lemma \ref{Lemma} states that, despite the lack of interaction between
the TE and defense modules, the target will still stand-out from the
remaining nodes, with at least two times more appearances in attack-affected
areas than the next suspect, \emph{in the worst case}. The probability
reaches zero when the attack takes place near the network leaves (see
Appendix). The upper bound corresponds to attacks closer to the gateways.
In both cases, however, the \emph{cumulative} appearances of nodes
within $\mathbf{Dst}^{t}$ favor the exposure of the target. Similarly,
cumulative appearances of IPs in flooded links favor in principle
the exposure of bots.

Based on these remarks, we define a TE-agnostic defense workflow,
formulated as Algorithm~\ref{ALG}. The workflow is event-based,
defining an initialization event and a TE event. The initialization
takes place once, at the defense module setup phase. Any routing instance
$r_{TE}\in\mathcal{R}_{TE}$ can be active in the network at this
stage. The TE event is triggered when $flooded(l)$ yields true for
any network link(s), which is the common case. The TE event then calls
a reinforcement learning module to update the probable bots and targets,
and executes the underlying TE scheme. The reinforcement learning
module is called in a non-blocking fashion, therefore causing no delay
or other overhead to the TE process.

A generic template for the reinforcement learning module is outlined
in lines $8-16$. It requires two standard hash tables (\ie $O(1)$
average complexity for insert, search and delete operations) for storing
the needed state. The two hash tables assign ``penalties'' to traffic
sources and to network nodes. Penalizing an IP ($penalize\_IPs$)
has the generic meaning of gradually collecting evidence of its malevolence,
taking into account its persistence (i.e., after re-routing). On the
other hand, penalizing a node ($penalize\_node$) corresponds to the
generic process of gradually gathering evidence of an attack being
launched against it (\eg flooding its connecting links). If a custom
criterion ($penalization\_conclusive$) yields that the penalization
process is conclusive (\eg penalties surpass a threshold set by the
network administrator), the network administrator can enforce an attack
mitigation policy of his choice. For instance, blacklisting can be
applied to the deduced bots, while extra paths can be deployed between
the network gateways and the attack target. Examples of instantiating
the mentioned abstract functions and mitigation policies are given
in Section \ref{seceval}.

We should note that the modeled workflow does not require any additional
input parameters, apart from the ones already gathered by most network
operators. Particularly, network routes, and subsequently $\mathbf{Dst}\left\{ l\right\} $,
are the \emph{needed} output of any TE scheme, which is critical to
the operation of any network. In addition, the monitoring of flows
over each route, and subsequently $\mathbf{Src}\left\{ l\right\} $,
is \emph{required} by standard defense mechanisms for commonplace
attacks~\cite{hoque2014network}. A trivial example is the mitigation
of link-flooding due to direct attacks with elephant flows. Moreover,
the asynchronous call mode of the learning module ensures that the
defense module runs in parallel to the TE objectives, without altering
or obstructing its operation. Finally, the complexity of the discussed
defense workflow is determined by the chosen reinforcement learning
module.

\subsection{Defense with Attack-aware TE\label{sec:remote}}

While an attack-unaware TE process allows for the exposure of the
attack, potentially more can be achieved by incorporating an attack-aware
TE module to the modeled workflow. Particularly, a classic, load-balancing
TE will remain oblivious of the attack and overlook the continuous
routing changes required to migrate from $r_{TE}^{t}$ to $r_{TE}^{t+1}$.
In cases of repeated and frequent migrations, such as in DDoS attacks,
these changes can be expensive in terms of routing table space and
disruptions caused to applications~\cite{godfrey2015stabilizing}.
Most importantly, the attacker is inherently allowed to affect the
routing of potentially the whole network. Therefore, the posed question
is whether an attack-aware TE module can limit the routing changes
and contain the disruptions, while still exposing the attack.

First, we employ the routing disruption metric $\mathcal{M}(t)$~\cite{Thaler.2000e}:
\begin{equation}
\mathcal{M}(t)=\sideset{}{_{\tau=1}^{t}}\sum r_{TE}^{\tau+1}\,\ominus\,r_{TE}^{\tau}\label{EQDISRUPT}
\end{equation}
Then, the novel ReMOTE (Routing MOdification minimizing TE) limits
$\mathcal{M}$ by \first reusing big parts of $r_{TE}^{t}$, effectively
retaining some memory of past routing paths, and \second affecting
potential target areas only. ReMOTE is essentially a binary search
for the target over several attack cycles, while ensuring that all
flooded links are relieved, i.e., their loads remain below a given
threshold. ReMOTE comprises two steps, executed at each cycle:\\
\quad{}$\text{\ding{182}}$. Bisect the routing tree defined by node
set $\mathbf{Dst}^{t}$ and $r_{TE}^{t}$, e.g. by using the process
of Reed et al.~\cite{reed1993finding} which has $O(\left\Vert \mathbf{Dst}^{t}\right\Vert )$
complexity. Thus, $\mathbf{Dst}^{t+1}$ will comprise half of the
$\mathbf{Dst}^{t}$ set.\\
\quad{}$\text{\ding{183}}$. Rehome each of the ensuing routing trees
to any gateway $g\in G$, using link-disjoint paths that can support
the added traffic load. This is achieved with a path finding algorithm
(e.g., Dijkstra or A$^{*}$) for $\overrightarrow{p_{g,d(l)}}$, sequentially
for each $flooded$ link $l$ and a gateway $g\in G$.

An example of ReMOTE is given in Fig. \ref{FIGEXAMPLES}, where the
attacker aims to disconnect node $7$. He thus sends load to other
decoy nodes and tries to cut the link $l_{g,5}$, denoted in bold.
The nodes in Fig.~\ref{FIGEXAMPLES} are annotated with their traffic
demand ratio at this point. Once the link flooding is detected, ReMOTE
applies step 1, bisecting the original routing tree comprising nodes
1-9. The tree bisection process seeks to define two routing subtrees
that are approximately equal both in terms of total traffic demand
and number of nodes. In the context of Fig.~\ref{FIGEXAMPLES}, the
bisection process returns the subtrees comprising nodes 1-5 and 6-9,
each yielding 4 nodes and 0.5 total traffic demand.

Subsequently, step 2 of ReMOTE is executed, and the sub-tree with
nodes 6-9 is routed to the gateway via a new path $g\to6$ (dashed
line, Fig.~\ref{FIGEXAMPLES}) which is link-disjoint to the previous,
$g\to5\to6$. The new path is selected so as to handle the added traffic
demand of nodes 6-9 (i.e., 0.5) without causing new flooding events.
At the next attack cycle, ReMOTE will focus on nodes 7-8, etc., eventually
pinpointing the target.

Finally, we collectively handle all exceptional cases, occurring,
e.g., when the maximum link utilization remains high after TE, or
when step 1 of ReMOTE returns empty sets. In any exceptional case,
ReMOTE hands over operation to a pure load-balancing TE~\cite{Buriol.2005},
for the current cycle only.
\begin{figure}[t]
\centering{}\includegraphics[bb=0bp 0bp 365bp 160bp,clip,width=0.95\columnwidth]{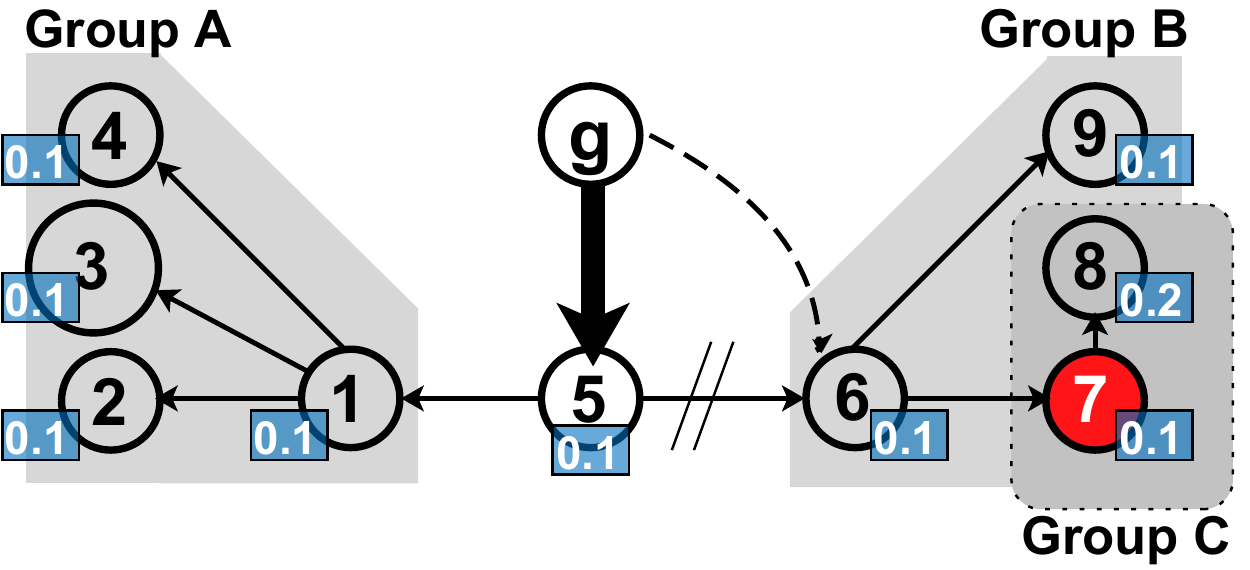}\vspace{-10bp}
 \caption{\label{FIGEXAMPLES} An example of ReMOTE in practice. The link $g,5$
is flooded by an attacker, and ReMOTE replies by isolating nodes 1-5
from 6-9 while also balancing their annotated load. Subsequent attacks
lead to the isolation of nodes (7, 8) from (6, 9), closing in on the
target (7). \vspace{-15bp}
}
\end{figure}

In the described manner, the attack is in principle contained within
a progressively-shrinking node set. Likewise, the defender's actions
will disrupt smaller parts of the network. As a binary search approach,
ReMOTE is completed at time $\tau\approx log_{2}\left\Vert \mathbf{Dst}^{1}\right\Vert $.
However, due to step (1), certain nodes will be persistently coupled
with the target. In particular, at time $\tau$, one node (i.e., $2^{1}-1$)
will have been coupled to the target $\tau-1$ times, $2^{2}-1$ nodes
$\tau-2$ times, and so on. Thus, in order to finally differentiate
between the target and suspect nodes, additional iterations are needed
as follows. The defender alters the specific routing of the most probable
target(s) (as of time $\tau$), deploying node-disjoint paths per
iteration. Thus, persistent couplings are broken and the target becomes
more discernible.

Notice that, counter-intuitively, an attack-unaware TE will generally
yield a more \emph{confident} target detection. According to Lemma~\ref{Lemma},
at an attack cycle $\tau$, the attack-unaware TE will have counted
$\tau$ participations of the real target in the affected node sets,
whereas any other node will have been observed just $\frac{\tau}{2}$
times. ReMOTE, on the other hand, will have counted $\tau-1$ participations
(out of $\tau$) for one non-targeted node, yielding a smaller detection
confidence. Thus, a possible trade-off between the two types of TE
is revealed. Particularly, an attack-unaware TE may discern the attack
target better than ReMOTE for the same number of iterations (i.e.,
faster detection). On the other hand, an attack-unaware TE can alter
the routing of the complete network, as opposed to the logarithmically
shrinking area-of-effect of ReMOTE. This trade-off is evaluated in
Section~\ref{seceval}.

\subsection{Collaboration with other defense schemes\label{sub:Collaboration}}

The effective mitigation of Crossfire requires collaboration among
existing defenses, especially when countering large botnets, as mentioned
in Section~\ref{sec:defense}. Within an Autonomous System (AS),
for example, existing solutions can force the attacker to use more
bots per attack iteration, facilitating their exposure~\cite{crossfire}.
Particularly, defenses based on flow rate monitoring can force the
attacker to use less traffic per bot and, hence, more bots to achieve
the same impact~\cite{braga2010lightweight}. Packet inspectors and
other heuristics can also peel-off the botnet by detecting and blocking
malevolent traffic sources independently. Thus, the bots at the attacker's
disposal decrease, forcing him to increase the reuse rate of the remaining
ones~\cite{xue2014towards}. For instance, IP traceback and TTL inspectors
can detect the origins of spoofed packets~\cite{ciscoSpoof}. Bot-Hunters
can detect bots based on the similarity of their traffic patterns~\cite{shin2013fresco}.
Phantom Nets can mislead an attacker into producing a false topology
of the network, while White-Holes can disguise real network nodes
as honeypots, tricking the attacker into an inefficient use of his
bots~\cite{shin2013fresco}. Such solutions can run independently,
forming a defense stack, while not being obstructed by (or obstruct)
the TE~module.

TE needs to trigger and work closely with inter-AS defenses, in the
case where a Crossfire attack has e.g., flooded all links around the
gateways. Thus, network-internal TE will be unable to load-balance
the traffic, calling for external help. TE is then applied in synergy
with the surrounding ASes~\cite{codef}. The AS-internal defense
stack can still run in parallel, removing as many bots as possible
in the process.

\section{Experimental Evaluation}

\label{seceval}

The defense workflow is implemented in the AnyLogic simulation platform~\cite{XJTechnologies.2015},
which offers visual, multi-paradigm programming and debugging, while
automating run repetitions to achieve a user-supplied confidence level
(set to $95\%$ for the present simulations). We incorporate an attack-unaware
TE~\cite{Buriol.2005} to Algorithm~\ref{ALG}, to validate the
analytical findings (Section~\ref{sec:def-workflow}). This TE is
representative of the load-balancing class of algorithms. It uses
path optimization driven by a Genetic Algorithm to achieve a classic
min-max link load utilization. We denote this approach as GATE. We
then incorporate ReMOTE to Algorithm~\ref{ALG}, in order to evaluate
the prospects of attack-aware TE (Section~\ref{sec:remote}). The
attacker model is as detailed in Section~\ref{sec:prereq}.

\textbf{\uline{Setup}}\textbf{.} The simulations assume $e=10,000$
IPs in total, which reside beyond the gateways. Out of these, $B_{size}$
(\%) are malevolent. The botnet of size $(B_{size}\cdot e)$ uses
$B_{part}$ (\%) randomly selected bots at each attack cycle, in order
to hinder detection. At the same time, a $P_{rehome}$ (\%) fraction
of the benign $\left(1-B_{size}\right)\cdot e$ hosts randomly picks
a new destination within the network, accentuating the appearance
of benign link-flooding traffic as well. We use several real topologies
from the TopologyZoo~\cite{topologyzoo} (full list in Fig.~\ref{fig:topos}),
and set the nodes in each topology with the northern and southern-most
geo-coordinates as gateway and target, respectively. The bandwidth
of each link is set to $1\,Gbps$. Each host, benign or bot, generates
up to $100\,Kbps$, indistinguishably. A link is \emph{flooded} if
its load exceeded $90\%$ at the past timestep.

\textbf{\uline{Penalization}}\textbf{.} We use the algorithm of
Misra et al. for the reinforcement learning, which runs on linear
complexity over the $\mathbf{Src}$, $\mathbf{Dst}$ data structures
and follows an event-counting logic~\cite{Misra.1982}. A node (potential
target) receives a $+1$ penalty each time its connectivity to the
gateway is affected by a link-flooding event. Likewise, an IP (potential
bot) is penalized by either $+1$ if it is found sending traffic over
a flooded link, or by $+2$ if it has \emph{also} changed its previous
destination to contribute \emph{again} to the link-flooding event.
If these criteria are not met, all non-zero penalties of nodes/IPs
are reduced by one. The penalization is deemed \emph{conclusive} at
timestep $\tau$, when the set of the most penalized nodes has not
changed for the last $5$ iterations, for stable detection. Then,
a host is classified as part of the botnet if it has accumulated more
than $\tau+1$ penalty points. The bot detection is more efficient
when $B_{part}$ is near $100\%$, since there will be at least one
penalization per bot per iteration.

\begin{table}[t]
\label{tbl:simulations}

\begin{centering}
\begin{tabular}{l|c|c|c||l|l|l|l|l|l|l}
\cline{2-10}
\multirow{2}{*}{\begin{turn}{90}
{\scriptsize{}Scenario~~~~~}
\end{turn}} & \multicolumn{3}{c|||}{{\scriptsize{}Parameters ($_{\%}$)}} & \multirow{2}{*}{\begin{turn}{90}
\parbox[c][1\totalheight][s]{45bp}{%
{\scriptsize{}Link util. \%}%
}
\end{turn}} & \multirow{2}{*}{\begin{turn}{90}
{\scriptsize{}Timesteps~~~~}
\end{turn}} & \multirow{2}{*}{\begin{turn}{90}
\parbox[c][1\totalheight][s]{45bp}{%
{\scriptsize{}Mod./node~~~}%
}
\end{turn}} & \multirow{2}{*}{\begin{turn}{90}
{\scriptsize{}Detect rate \%}
\end{turn}} & \multirow{2}{*}{\begin{turn}{90}
{\scriptsize{}False pos. \%}
\end{turn}} & \multirow{2}{*}{\begin{turn}{90}
\parbox[c][1\totalheight][s]{48bp}{%
{\scriptsize{}Tgt. detect \%}%
}
\end{turn}} & \multirow{2}{*}{\begin{turn}{90}
{\scriptsize{}TE scheme$^{*}$}
\end{turn}}\tabularnewline
\cline{2-4}
 & \begin{turn}{90}
{\scriptsize{}B~size}
\end{turn} & \begin{turn}{90}
 {\scriptsize{}B~part}
\end{turn} & \begin{turn}{90}
{\scriptsize{}P~rehome~}
\end{turn} &  &  &  &  &  &  & \tabularnewline
\cline{2-10}
\multirow{2}{*}{\textbf{\scriptsize{}1}} & \multirow{2}{*}{\textbf{\scriptsize{}10}} & \multirow{2}{*}{\textbf{\scriptsize{}100}} & \multirow{2}{*}{\textbf{\scriptsize{}0}{\scriptsize{} }} & \textbf{\scriptsize{}75}{\scriptsize{} } & \textbf{\scriptsize{}14}{\scriptsize{} } & \textbf{\scriptsize{}13}{\scriptsize{} } & \textbf{\scriptsize{}95}{\scriptsize{} } & \textbf{\scriptsize{}0}{\scriptsize{} } & \textbf{\scriptsize{}95}{\scriptsize{} } & {\tiny{}R}\tabularnewline
\cline{5-10}
 &  &  &  & \textbf{\scriptsize{}65}{\scriptsize{} } & \textbf{\scriptsize{}12}{\scriptsize{} } & \textbf{\scriptsize{}58}{\scriptsize{} } & \textbf{\scriptsize{}99}{\scriptsize{} } & \textbf{\scriptsize{}0}{\scriptsize{} } & \textbf{\scriptsize{}99}{\scriptsize{} } & {\tiny{}G}\tabularnewline
\cline{2-10}
\multirow{2}{*}{{\scriptsize{}$_{a}$}} & \multirow{2}{*}{{\scriptsize{}10}} & \multirow{2}{*}{{\scriptsize{}100}} & \multirow{2}{*}{{\scriptsize{}100} } & {\scriptsize{}75}  & {\scriptsize{}14}  & {\scriptsize{}13}  & {\scriptsize{}95}  & {\scriptsize{}0}  & {\scriptsize{}95}  & {\tiny{}R}\tabularnewline
\cline{5-10}
 &  &  &  & {\scriptsize{}65}  & {\scriptsize{}12}  & {\scriptsize{}55}  & {\scriptsize{}95}  & {\scriptsize{}0}  & {\scriptsize{}99}  & {\tiny{}G}\tabularnewline
\cline{2-10}
\multirow{2}{*}{{\scriptsize{}$_{b}$}} & \multirow{2}{*}{{\scriptsize{}10} } & \multirow{2}{*}{{\scriptsize{}50} } & \multirow{2}{*}{{\scriptsize{}100}} & {\scriptsize{}75}  & {\scriptsize{}14}  & {\scriptsize{}13}  & {\scriptsize{}25}  & {\scriptsize{}0}  & {\scriptsize{}95}  & {\tiny{}R}\tabularnewline
\cline{5-10}
 &  &  &  & {\scriptsize{}65}  & {\scriptsize{}12}  & {\scriptsize{}59}  & {\scriptsize{}29}  & {\scriptsize{}0}  & {\scriptsize{}99}  & {\tiny{}G}\tabularnewline
\cline{2-10}
\multirow{2}{*}{{\scriptsize{}$_{c}$} } & \multirow{2}{*}{{\scriptsize{}10} } & \multirow{2}{*}{{\scriptsize{}10} } & \multirow{2}{*}{{\scriptsize{}100}} & {\scriptsize{}76}  & {\scriptsize{}14}  & {\scriptsize{}13}  & {\scriptsize{}11}  & {\scriptsize{}0}  & {\scriptsize{}95}  & {\tiny{}R}\tabularnewline
\cline{5-10}
 &  &  &  & {\scriptsize{}67}  & {\scriptsize{}12}  & {\scriptsize{}57}  & {\scriptsize{}10}  & {\scriptsize{}0}  & {\scriptsize{}99}  & {\tiny{}G}\tabularnewline
\cline{2-10}
\multirow{2}{*}{\textbf{\scriptsize{}2}} & \multirow{2}{*}{\textbf{\scriptsize{}50}{\scriptsize{} }} & \multirow{2}{*}{\textbf{\scriptsize{}100}} & \multirow{2}{*}{\textbf{\scriptsize{}0}} & \textbf{\scriptsize{}72}{\scriptsize{} } & \textbf{\scriptsize{}14}{\scriptsize{} } & \textbf{\scriptsize{}11}{\scriptsize{} } & \textbf{\scriptsize{}98}{\scriptsize{} } & \textbf{\scriptsize{}0}{\scriptsize{} } & \textbf{\scriptsize{}95}{\scriptsize{} } & {\tiny{}R}\tabularnewline
\cline{5-10}
 &  &  &  & \textbf{\scriptsize{}57}{\scriptsize{} } & \textbf{\scriptsize{}12}{\scriptsize{} } & \textbf{\scriptsize{}56}{\scriptsize{} } & \textbf{\scriptsize{}99}{\scriptsize{} } & \textbf{\scriptsize{}0}{\scriptsize{} } & \textbf{\scriptsize{}99}{\scriptsize{} } & {\tiny{}G}\tabularnewline
\cline{2-10}
\multirow{2}{*}{\textbf{\scriptsize{}3}} & \multirow{2}{*}{\textbf{\scriptsize{}90}} & \multirow{2}{*}{\textbf{\scriptsize{}100}{\scriptsize{} }} & \multirow{2}{*}{\textbf{\scriptsize{}0}} & \textbf{\scriptsize{}69}{\scriptsize{} } & \textbf{\scriptsize{}14}{\scriptsize{} } & \textbf{\scriptsize{}12}{\scriptsize{} } & \textbf{\scriptsize{}98}{\scriptsize{} } & \textbf{\scriptsize{}0}{\scriptsize{} } & \textbf{\scriptsize{}95}{\scriptsize{} } & {\tiny{}R}\tabularnewline
\cline{5-10}
 &  &  &  & \textbf{\scriptsize{}50}{\scriptsize{} } & \textbf{\scriptsize{}12}{\scriptsize{} } & \textbf{\scriptsize{}58}{\scriptsize{} } & \textbf{\scriptsize{}99}{\scriptsize{} } & \textbf{\scriptsize{}0}{\scriptsize{} } & \textbf{\scriptsize{}99}{\scriptsize{} } & {\tiny{}G}\tabularnewline
\cline{2-10}
\end{tabular}
\par\end{centering}

\caption{\label{tab:Evaluation}Evaluation metrics \emph{averaged over all
topologies} using \uline{R}eMOTE$^{*}$ and \uline{G}ATE$^{*}$
under different parameter settings.\vspace{-12bp}
}
\end{table}

\textbf{\uline{Simulation Results}}\textbf{.} Table \ref{tab:Evaluation}
shows six evaluation metrics averaged over all topologies and simulation
runs for selected parameters. We first observe that the TE choice
yields the expected trade-off between detection speed/rate and the
number of routing rule modifications needed (Section~\ref{sec:remote}).
The routing of GATE results in slightly fewer iterations and higher
detection rates. Nonetheless, ReMOTE achieves less routing changes
per node. Besides, the average bot detection rate is not significantly
affected by botnet size variations (scenarios 1, 2 and 3). Scenario
1.a sets $P_{rehome}=100\%$, accentuating the appearance of random
flood events attributed to benign hosts. However, random flood events
have no persistent target area in general. In contrast, malicious
flooding has a persistent target and can be accurately detected. Scenarios
1.b and 1.c assume an attacker that seeks to avoid detection by decreasing
the $B_{part}$ ratio. As expected, the bot detection rate drops considerably~\cite{crossfire}.
The detection rates of the attack target are promising for both schemes,
but GATE has the advantage, as explained in Section~\ref{sec:remote}.
GATE achieves lower link utilization after TE than ReMOTE, given that
it constitutes its sole optimization objective.

On the other hand, ReMOTE requires far fewer routing changes than
GATE. Fig.~\ref{fig:topos} shows boxplots of the number of modifications
per node for ReMOTE and GATE. Furthermore, the variance of this number
of changes is very low for ReMOTE, given that it contains the attack
within rapidly shrinking network areas.

\begin{figure}[t]
\centering{} \includegraphics[width=1\columnwidth,height=8cm]{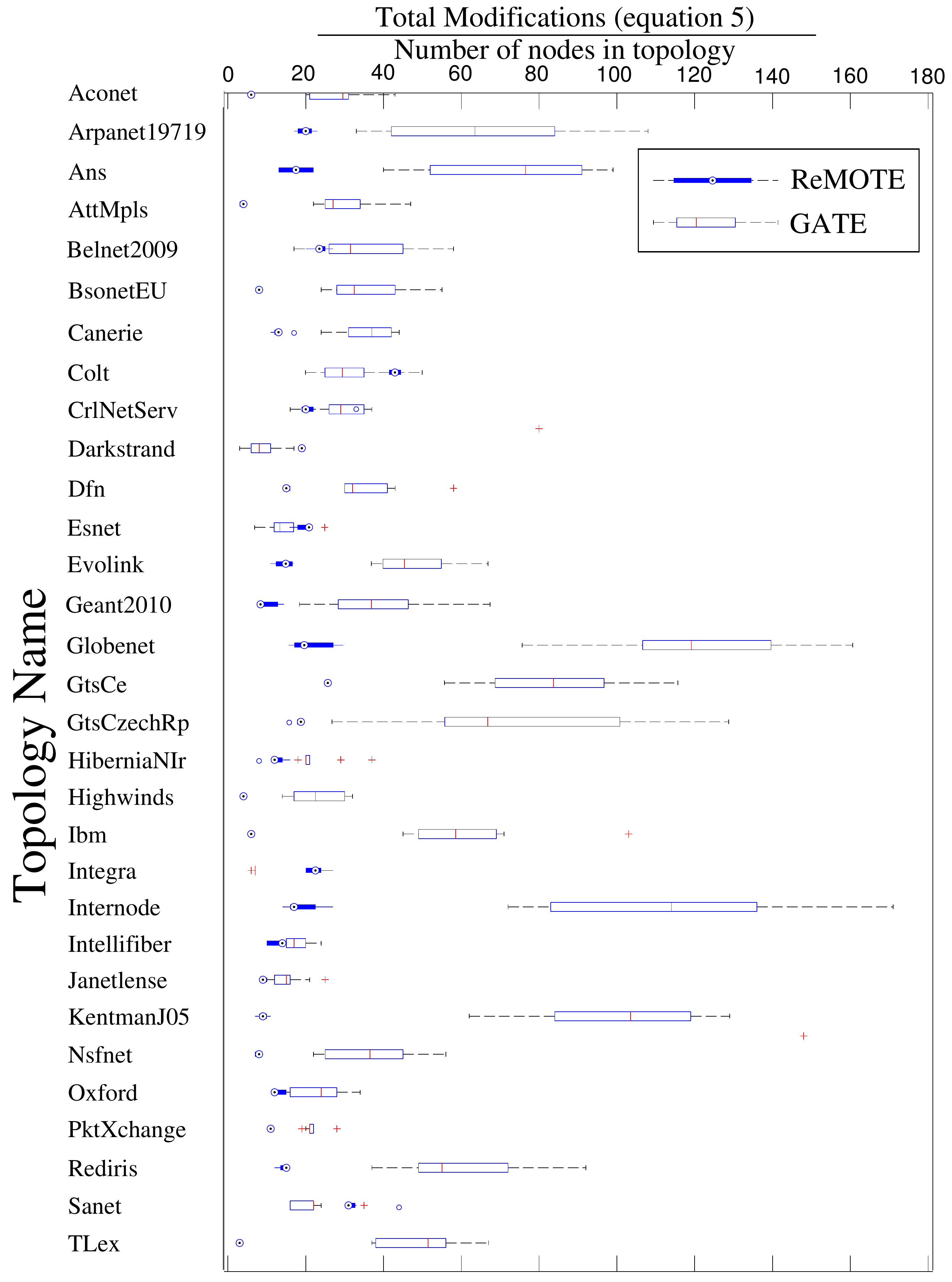}
\caption{\label{fig:topos}Boxplots of the number of rule modifications per
node with ReMOTE and GATE for several real topologies. \vspace{-8bp}
}
\end{figure}

\section{Related Work\label{SECRW}}

DDoS attacks have attracted notable research interest, particularly
in recent years. Braga \etal~\cite{braga2010lightweight} presented
a low-overhead technique for traffic analysis using Self Organizing
Maps (SOMs) to classify flows and enable DDoS attack detection caused
by \emph{direct }heavy hitters. Ashraf \etal~\cite{ashraf2014handling}
study several machine learning approaches for use in counter-DDoS
heuristics. Hommes \etal~investigate DoS attacks from the aspect
of \emph{routing table} space depletion~\cite{hommes2014implications}.
Lim \etal~propose a scheme to block botnet-based DDoS attacks that
do not exhibit detectable statistical anomalies~\cite{lim2014sdn}.
They focus on HTTP communications between clients and servers and
they employ \emph{CAPTCHA challenges} for HTTP URL redirection. Besides,
Xue \etal~\cite{xue2014towards} propose \emph{LinkScope} to detect
new classes of link-flooding attacks and locate critical links whenever
possible. \emph{LinkScope} employs end-to-end and hop-by-hop network
measurement techniques to capture abrupt performance degradation,
while \emph{packet inspection} is required. The work of Lee \etal~(CoDef)~\cite{codef}
is an \emph{inter-AS} approach towards defeating attacks such as Coremelt~\cite{Coremelt}
and Crossfire~\cite{crossfire}. CoDef proposes a cooperative method
for identifying low-rate attack traffic. Nonetheless, the emphasis
is on the communication among distrustful autonomous systems, and
not on the role of the TE process. We refer the reader to the survey
of Zargar \etal~\cite{zargar2013survey} for a comprehensive overview
of defenses against link flooding attacks.

The present paper differentiates by studying the role of TE in Crossfire
detection and mitigation. We note, however, that the presented workflow
can coexist with related solutions as well (cf. Section~\ref{sub:Collaboration}).
The authors' prior work studied the use of TE against Crossfire in
networks that employ \emph{source-based} routing and \emph{static
paths} \cite{INFOCOM}. The present work refers to \emph{destination-based}
routing and \emph{variable paths}. In this context, it contributes
a practical workflow for Crossfire mitigation, which can be used in
conjunction with existing TE, as well as with novel attack-aware TE
schemes.

\section{Conclusion\label{SECCONCLUSION}}

The present study showed that a class of stealthy link-flooding attacks
can be exposed by exploiting existing and novel TE schemes. A generic
defense workflow was formulated, which can be used efficiently in
conjunction even with common TE algorithms. A trade-off was indicated,
showing that there exist TE schemes that contain the attack within
isolated network areas, at the cost of slower exposure. Optimizing
this trade-off is the objective of future work.

{\scriptsize{}\bibliographystyle{acm}

}\vspace{-15bp}

\section*{Appendix}

\textbf{Proof of Lemma \ref{Lemma}.} Consider a physical topology
and all possible routing tree choices $\mathcal{R}_{TE}$. We will
assume that for every possible subset $\mathcal{S}$ of the topology
nodes, there exists a link $l$ and a routing tree $r_{TE}^{t}\in\mathcal{R}_{TE}$
such as $\mbox{\textbf{Dst}}\left\{ l\right\} =\mathcal{\mathcal{S}}$.
In other words, the attack can potentially focus on any node subset.
In addition, let the attacker be able to flood links that serve $k$
or less nodes. The variable $k$ qualitatively expresses the severity
of the attack. Attacks against network leaves affect few nodes (low
$k$), while attacks near, e.g., a gateway affect great parts of the
network (higher $k$).

The number of all possible subsets of $\mathcal{N}$, with size \uline{up
to} $k$, that contain one given target is $g_{k}=\sum_{m=1}^{k}\binom{\left\Vert \mathcal{N}\right\Vert -1}{m-1}$.
Any non-targeted node is contained in such an $m$-sized set with
probability $\frac{m-1}{\left\Vert \mathcal{N}\right\Vert -1}$. Thus,
the probability of a non-targeted node being coupled to the target
in groups of up to $k$ nodes is $p_{k}=\frac{f_{k}}{g_{k}}$, where
$f{}_{k}=\sum_{m=1}^{k}\frac{m-1}{\left\Vert \mathcal{N}\right\Vert -1}\binom{\left\Vert \mathcal{N}\right\Vert -1}{m-1}=\sum_{m=1}^{k}\binom{\left\Vert \mathcal{N}\right\Vert -2}{m-2}$.

Finally, $p_{1}=0$ and $p_{\left\Vert \mathcal{N}\right\Vert }\approx\frac{1}{2}$
due to the binomial theorem, while $p_{k}$ is strictly rising:\\
Let $\Delta p=p_{k+1}-p_{k}$. Then, via finite differences we write:\vspace{-4bp}
\begin{equation}
\Delta p=\frac{\Delta f\cdot g_{k}-f_{k}\cdot\Delta g}{g{}_{k}\cdot\left(g_{k}+\Delta g\right)}\mbox{,}\,\Delta g={\scriptstyle \binom{\left\Vert \mathcal{N}\right\Vert -1}{k}},\Delta f={\scriptstyle \binom{\left\Vert \mathcal{N}\right\Vert -2}{k-1}}
\end{equation}
Notice that $\frac{\Delta g}{\Delta f}=\frac{\left\Vert \mathcal{N}\right\Vert -1}{k}$.
Thus, $\Delta p\propto\Delta f\left(g_{k}-f_{k}\frac{\left\Vert \mathcal{N}\right\Vert -1}{k}\right)$.\\
In addition, $f{}_{k}=\sum_{m=1}^{k}\frac{m-1}{\left\Vert \mathcal{N}\right\Vert -1}\binom{\left\Vert \mathcal{N}\right\Vert -1}{m-1}$.
Therefore:\\
$\Delta p\propto\sum_{m=1}^{k}\left(\frac{k-m+1}{k}\right)\binom{\left\Vert \mathcal{N}\right\Vert -1}{m-1}>0$,
hence $p_{k+1}>p_{k}$, QED.$\blacksquare$
\end{document}